\documentclass{isipta2025}

\usepackage{amsmath}    
\usepackage{tikz}       
\usepackage{booktabs}   
\usepackage{fvextra}    
\usepackage{array}      
\usepackage{tabularx}   
\usepackage{enumitem}   

\title{Function-Coherent Gambles with Non-Additive Sequential Dynamics}
\author[1]{Gregory Wheeler}
\affil[1]{Frankfurt School of Finance \& Management\\Germany}

\begin{document}
\maketitle


\begin{abstract}
The desirable gambles framework provides a rigorous foundation for imprecise probability theory but relies heavily on linear utility via its coherence axioms. In our related work, we introduced function-coherent gambles to accommodate non-linear utility. However, when repeated gambles are played over time---especially in intertemporal choice where rewards compound multiplicatively---the standard additive combination axiom fails to capture the appropriate long-run evaluation. In this paper we extend the framework by relaxing the additive combination axiom and introducing a nonlinear combination operator that effectively  aggregates repeated gambles in the log-domain. This operator preserves the time-average (geometric) growth rate and addresses the ergodicity problem. We prove the key algebraic properties of the operator, discuss its impact on coherence, risk assessment, and representation, and provide a series of illustrative examples. Our approach bridges the gap between expectation values and time averages and unifies normative theory with empirically observed non-stationary reward dynamics.
\end{abstract}


\section{Introduction}
The desirable gambles framework \cite{Williams:1975,Walley:2000,DeCoomanQuaeghebeur:2012,LowerPrevisions,Wheeler:2022,DeBock:2023,deCoomanVanCampDeBock:2023} is the elemental core of a general account of imprecise probabilities, subsuming lower  \cite{Walley:1991,LowerPrevisions} and linear previsions \cite{DeFinetti:1974} as special cases. Yet, its coherence axioms (in particular, the combination axiom) encode a commitment to linear utility. Although  analytically convenient, linear utility fails to capture phenomena observed in intertemporal choice and repeated gambles where rewards are compounded multiplicatively. In such contexts, the arithmetic sum of outcomes is not the appropriate aggregator; rather, the long-run performance is determined by the geometric mean or time-average growth rate.

Previous work \cite{Wheeler:2021.isipta} introduced \emph{discounted desirable games} to relax the linearity assumption while preserving essential rationality conditions \cite{Wheeler:2025.func}. Here, we extend that approach by making three primary contributions. First, we develop a novel combination operator that preserves coherence while accommodating non-linear utility, addressing a fundamental limitation in the standard desirable gambles framework. This operator, defined as $f \oplus g = (1 + f)(1 + g) - 1$, naturally captures the multiplicative dynamics of compound growth while maintaining essential rationality properties.

Second, we establish necessary and sufficient conditions for well-behaved combination operators through a representation theorem. Our characterization reveals that the logarithmic transformation is not merely one among many possible choices, but emerges naturally as the unique operator that simultaneously preserves function coherence in the transformed space, maintains the time-average geometric growth rate, and ensures additivity of sequential risks in the log-domain.

Third, we unify several seemingly disparate concepts within a coherent mathematical framework. We demonstrate how the ergodicity problem in multiplicative dynamics, function-coherent representations of risk preferences, time-average growth optimization, and non-stationary reward processes are fundamentally connected through the structure of our combination operator and its induced risk measure. This unification provides new insights into the relationship between expectation values and time averages, bridging the gap between normative decision theory and empirically observed behavior in dynamic choice situations.

These theoretical results are complemented by practical applications, particularly in portfolio management and long-horizon decision problems, where we show how our framework naturally captures phenomena such as volatility drag and the asymmetric impact of gains and losses---effects that often require ad hoc adjustments in traditional approaches.

\section{Preliminaries and Motivation}

The desirable gambles framework consists of axioms for constructing a coherent set $\mathbb{D}$ of bounded gambles. For bounded gambles $f, g$ and positive real number $\lambda \geq 0$, a set of bounded gambles $\mathbb{D}$ is coherent when satisfying:

\begin{itemize}
\item[A1.]  If $f < 0$, then $f \not\in \mathbb{D}$ \hfill ({\bf Avoid partial loss})
\item[A2.] If $f \geq 0$, then $f \in \mathbb{D}$  \hfill ({\bf Accept partial gain}) 
\item[A3.] If $f \in \mathbb{D}$, then $\lambda f \in \mathbb{D}$ \hfill ({\bf Pos.~Scale Invariance})
\item[A4.]  If $f \in \mathbb{D}$ and $g \in \mathbb{D}$, then $f + g \in \mathbb{D}$  \hfill ({\bf Combination})
\end{itemize}

Axioms A1 and A2 express core rationality conditions, while A3 and A4 are closure operations encoding linear utility. Together they define a convex cone containing all conic combinations of its elements:


\begin{equation}
\label{eq:conic_hull}
\mathsf{cone}(\mathbb{D}) := \left\{ \sum_{i =1}^{n} \lambda_i f_i :  f_i\in \mathbb{D}, i=1,\ldots,n,\; \lambda_i \geq 0   \right\}
\end{equation}

\subsection{The Sequential Decision Problem}

While this framework elegantly captures many aspects of decision making under uncertainty, it encounters limitations when applied to sequential decisions, particularly those involving multiplicative dynamics. Consider two key cases:

\begin{enumerate}
\item \emph{Additive Accumulation}: For a gamble that modifies wealth by adding $f(\omega)$ in state $\omega$, repeated application over $n$ periods yields total change:
\[
\sum_{i=1}^n f(\omega_i)
\]
This aligns naturally with axiom A4.

\item \emph{Multiplicative Growth}: For a gamble that modifies wealth by factor $(1+f(\omega))$, the $n$-period wealth evolution is:
\[
w_n = w_0\prod_{i=1}^n (1+f(\omega_i))
\]
Here, the additive combination axiom fails to capture the compound growth structure.
\end{enumerate}

\begin{example}
\label{example:seq-gamble}
To see the difference between an arithmetic‐mean return and long‐run (time‐average) growth, 
consider two investment options:

\begin{itemize}
\item Option A: Returns +50\% or -40\% with equal probability each period.
\item Option B: Consistently returns +5\% every period.
\end{itemize}

A single‐period, \emph{arithmetic} expected return analysis would tell us:
\[
  \mathbb{E}[\text{Option A}] = 0.5 \times (+50\%) + 0.5 \times (-40\%) \;=\; +5\%,
\]
the same as Option B's certain +5\%.  From this perspective, one might think the two 
investments are ``equally good.'' 

However, when invested repeatedly, wealth compounds.  Starting with \$100 in Option A:

\begin{itemize}
\item \textbf{Up then Down:} $100 \times 1.5 \times 0.6 = \$90$
\item \textbf{Down then Up:} $100 \times 0.6 \times 1.5 = \$90$
\item \textbf{Up then Up:} $100 \times 1.5 \times 1.5 = \$225$
\item \textbf{Down then Down:} $100 \times 0.6 \times 0.6 = \$36$
\end{itemize}

Meanwhile, investing repeatedly in Option B (at +5\% each period) grows \$100 to 
\$110.25 after two periods:
\[
  100 \;\times\; 1.05 \;\times\; 1.05 \;=\; \$110.25.
\]
Notice that in \emph{three out of these four} equally likely scenarios for Option A, 
the final wealth is less than \$110.25.  

This discrepancy arises because the \emph{expected value} of Option A for a single 
period (+5\%) does \emph{not} reflect how actual wealth evolves through time when 
gains and losses compound.  Indeed, the average (arithmetic) return overlooks 
the fact that recovering from a 40\% loss requires a greater‐than‐40\% gain. 
When decisions must be repeated many times, the \emph{time‐average} growth rate 
often diverges from the simple arithmetic mean of returns, which can lead 
to misjudgments if one only uses traditional expected‐value analysis.
\end{example}

This divergence between ensemble averages (what we expect across many parallel universes) and time averages (what a single investor actually experiences over time) is the essence of the ergodicity problem \cite{Birkhoff:1931,Peters:2016}. In multiplicative processes like investment returns, the sequence and path of outcomes matter fundamentally. The arithmetic average of returns fails to capture this path dependence, leading to potentially misleading evaluations of long-term growth prospects.

This distinction becomes crucial when evaluating long-run performance. For a gamble $f$ over $|\Omega|=m$ states, let the reward vector $\mathbf{x}_f = (x_1, \ldots, x_m)$ in $\mathbb{R}^m$ represent state-contingent outcomes. Under linear utility:

\begin{equation}
\label{eq:lin-U}
u_{lin}(\mathbf{x}) = \mathbf{1}\mathbf{x} = \mathbf{x}
\end{equation}

However, when outcomes compound multiplicatively, linear utility fails to capture the asymmetric impact of gains and losses on long-term growth.  Figure~\ref{fig:multi_2a} contrasts the ensemble average, $\mathbb{E}_p(f)$ under additive accumulation, with the time-average, multiplicative growth rate, $\mathbb{E}_t(f)$. The figure makes clear that, for multiplicative processes, these two measures can diverge substantially. Relying solely on the ensemble average can thus mislead decision makers about the true long-run performance of a gamble.

\begin{figure}[t]
\centering
\includegraphics[width=\columnwidth]{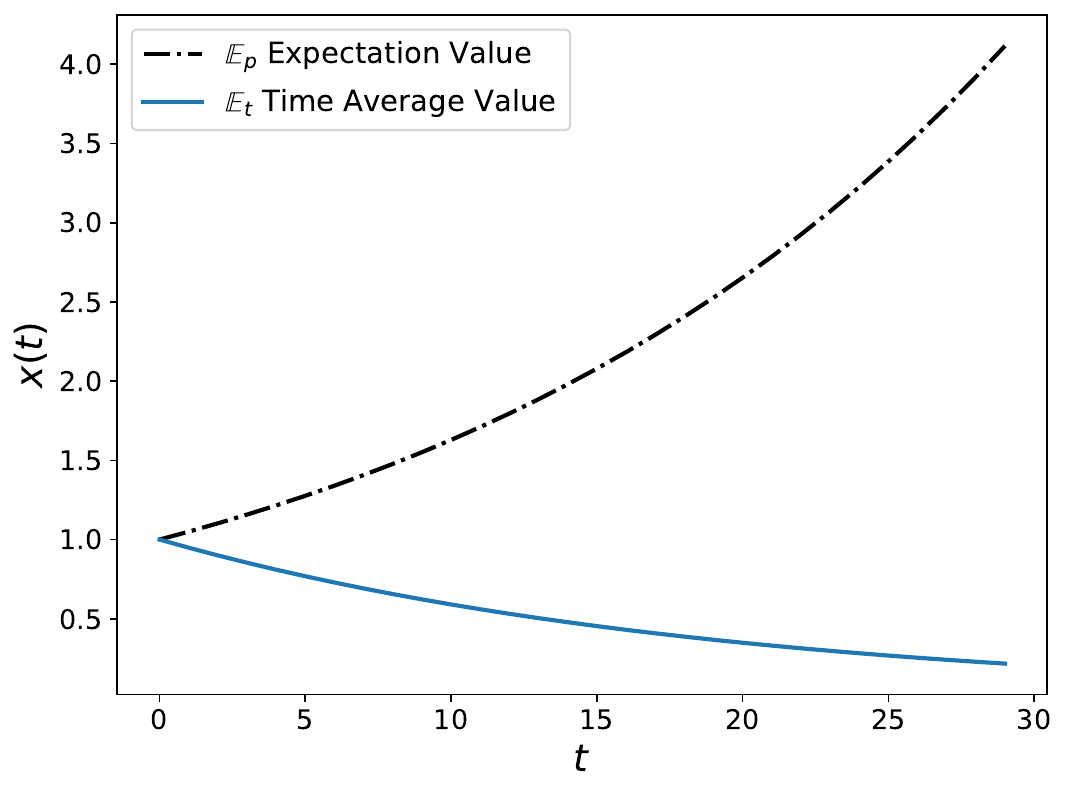}
\caption{Comparison of ensemble averaging versus time averaging of Option A in Example~\ref{example:seq-gamble} over a fixed period of 30 iterations. The dashed line represents the expectation value $\mathbb{E}_p(f)$ under additive accumulation, while the solid line depicts the time-average growth rate, $\mathbb{E}_t(f)$, under multiplicative growth. The divergence between these measures highlights the ergodicity problem in multiplicative processes.}
\label{fig:multi_2a}
\end{figure}

This motivates consideration of non-linear utility functions, starting with discounted utility  \cite{Wheeler:2021.isipta}:

\begin{definition}[Discounted Utility]
\begin{equation}
\label{eq:util_original}
u(x,\alpha) := \frac{x^{1-\alpha}-\alpha}{1-\alpha}, \quad \alpha \in [0,1),\; x > 0.
\end{equation}
\end{definition}

This discounted utility function takes a positive scalar reward $x$ and discounts its desirability to degree $\alpha$. When $\alpha = 0$, we recover linear utility $u(x,0) = x$. As $\alpha$ approaches 1, the function increasingly resembles logarithmic utility, providing a bridge to the multiplicative case.

\subsection{Challenge of Sequential Decisions}

The framework of desirable gambles with linear utility faces three fundamental challenges in sequential decision contexts:

\begin{enumerate}
\item \emph{Growth Rate Evaluation}: The arithmetic mean of returns fails to capture the geometric growth rate that determines long-run performance under multiplicative dynamics.

\item \emph{Risk Assessment}: Linear utility implies symmetric treatment of gains and losses, but multiplicative processes are inherently asymmetric—a 50\% loss requires a 100\% gain to recover.

\item \emph{Dynamic Consistency}: The standard combination axiom may lead to dynamically inconsistent preferences when applied to sequences of multiplicative gambles.
\end{enumerate}

These challenges motivate our development of a more general framework that can accommodate both additive and multiplicative dynamics while preserving the essential rationality properties of desirable gambles. In the following sections, we first introduce function-coherent gambles \cite{Wheeler:2025.func} to handle non-linear utility, then develop a theory of combination operators suitable for sequential decisions.

\section{Function-Coherent Gambles}

Having established why standard additive combination fails to capture multiplicative dynamics, we now introduce a more general framework that can accommodate non-linear utility while preserving essential rationality properties \cite{Wheeler:2025.func}. The key insight is to relax the linear utility assumption embedded in traditional desirable gambles while maintaining a convex structure in an appropriately transformed space.

\subsection{Basic Framework}

Following \cite{Wheeler:2025.func}, we start by generalizing the coherence axioms to accommodate non-linear utility functions. Given a strictly increasing, continuous utility function $u: X \to \mathbb{R}$ normalized by $u(0)=0$, a set of gambles \( \mathbb{D} \subset X \) is said to be \emph{function-coherent} if it satisfies the following axioms:

\begin{itemize}
\item[(F1)] \textbf{Avoid Losses:} If $f<0$ (a sure loss), then $f\notin \mathbb{D}$.
\item[(F2)] \textbf{Monotonicity:} If $f \ge g$ (pointwise) and $g\in \mathbb{D}$, then $f\in \mathbb{D}$.
\item[(F3)] \textbf{$u$-Convexity:} For any $f,g\in \mathbb{D}$ and any nonnegative scalars $\lambda,\mu$ for which 
\[
h = u^{-1}\bigl(\lambda\, u(f)+\mu\, u(g)\bigr)
\]
is well-defined, we have $h\in \mathbb{D}$.
\end{itemize}

The acceptance set is then naturally defined as
\[
\mathbb{D} = \{ f\in X: u(f)\ge 0 \}.
\]

Observe that if we choose $u$ to be the identity function, these axioms reduce to the classical desirability axioms.  Moreover, by defining the transformed acceptance set
\[
U(\mathbb{D}) = \{ u(f) : f \in \mathbb{D} \}.
\]
we see that $U(\mathbb{D})$ forms a convex cone. Thus, even though $u$ may be nonlinear, the mapping into the $u$-space preserves the essential convex (and hence linear) structure required for coherence.

\subsection{Representation Theorem}

Under appropriate regularity conditions, this convex structure leads to a  representation theorem that characterizes function-coherent sets of gambles. Following \cite{Wheeler:2025.func}, we require:
\begin{enumerate}[label=(\roman*)]
\item $U(\mathbb{D})$ is closed in the appropriate topology
\item $U(\mathbb{D})$ has nonempty interior
\end{enumerate}

\begin{theorem}[Representation Theorem]
\label{thm:rep}
Under the conditions above, there exists a nonzero continuous linear functional $\ell: V\to\mathbb{R}$, unique up to multiplication by a positive scalar, such that for every gamble $f\in X$,
\[
f\in \mathbb{D} \quad \Longleftrightarrow \quad \ell\bigl(u(f)\bigr) \ge 0.
\]
Equivalently, defining the evaluation functional
\[
\rho(f):=\ell\bigl(u(f)\bigr),
\]
we have
\[
f\in \mathbb{D} \quad \Longleftrightarrow \quad \rho(f)\ge 0.
\]
\end{theorem}

This representation theorem reveals that—even though gambles are evaluated through a nonlinear function $u$—the acceptance criterion can be represented linearly in the transformed space. When $\ell$ corresponds to integration with respect to a probability measure, we recover representations of the form
\[
f\in \mathbb{D} \quad \Longleftrightarrow \quad \mathbb{E}_p\bigl[u(f)\bigr]\ge 0,
\]
providing a bridge to classical expected utility theory while maintaining the flexibility of non-linear evaluation.

\subsection{Discussion and Domain Restriction}

A crucial aspect of the {\em function-coherence} framework is the domain restriction inherent in the utility function $u$. All gambles must lie in the domain $X$ where $u$ is well-defined. This is not merely a technical constraint but reflects the reality that not all possible gambles may be meaningfully evaluated under a given utility function.

For example, with logarithmic utility $u(x) = \log(1+x)$, the domain naturally excludes gambles that could lead to negative wealth. More generally, the domain restriction encodes structural assumptions about which gambles are economically meaningful in a given context.

The upward closure property (F2) remains well-behaved under this domain restriction.  Again from \cite{Wheeler:2025.func}, we have:

\begin{theorem}[Upward Closure Under Domain Restriction]
Let $f$ be an acceptable gamble (i.e., $f\in \mathbb{D}$) and let $g$ be any gamble in $X$ such that $g(s) \ge f(s)$ for every state $s$. Then $g\in \mathbb{D}$.
\end{theorem}

This framework sets the stage for developing non-additive combination rules that can properly capture multiplicative dynamics while maintaining coherence. The representation theorem will prove crucial in establishing that these new combination rules preserve the essential properties we desire.

\section{Non-Additive Sequential Dynamics}

Building on the function-coherent framework, we now develop a theory of combination operators that properly captures multiplicative dynamics while preserving coherence. The standard additive combination axiom (A4) from \cite{Williams:1975,LowerPrevisions} guarantees that if two gambles are acceptable, their arithmetic sum is also acceptable. While this aligns with linear utility, many real-world decision problems—particularly in domains like long-term investments studied in \cite{Wheeler:2021.isipta}—exhibit multiplicative dynamics where wealth evolves through compounding rather than simple addition.

To illustrate, consider a gamble $f$ that updates wealth by a factor of $1+f(\omega)$ in state $\omega$. After $n$ independent repetitions, wealth evolves as
\[
w' = w \prod_{i=1}^{n} \bigl(1+f(\omega_i)\bigr).
\]
Following \cite{Peters:2016,Peters:2019,Wheeler:2021.isipta}, we observe that the long-run performance of such a process is determined by the geometric mean, or equivalently, by the time-average of logarithmic returns:
\[
\frac{1}{n}\sum_{i=1}^{n}\log\bigl(1+f(\omega_i)\bigr).
\]

\subsection{The Nonlinear Combination Operator}

Motivated by this observation and building on \cite{Wheeler:2021.isipta}, we introduce a \emph{nonlinear combination operator} that respects multiplicative compounding by working in the logarithmic domain. 
Let $f$ and $g$ be gambles defined with respect to a state space $\Omega$ with the property that
\[ f(\omega) > -1 \quad \mathrm{and} \quad g(\omega) > -1 \quad \mathrm{for \ all} \ \omega \in \Omega \]
This condition ensures that the logarithmic transformation is well defined.  Next define the transformation function
\[
\phi(x)=\log(1+x),\quad x>-1
\] 
with its  inverse given by  
\[
\phi^{-1}(y)=e^y-1.
\]
Then, for any two gambles $f$ and $g$, we define their sequential (pointwise) combination:
\begin{itemize}
    \item[(F4)] \textbf{Nonlinear Combination:} If $f \in \mathbb{D}$ and $g \in \mathbb{D}$, then their nonlinear combination
    \begin{align*}
    f \oplus g(\omega) &:=  \phi^{-1}\Bigl(\phi\bigl(f(\omega)\bigr) + \phi\bigl(g(\omega)\bigr)\Bigr) \\
    		   &= (1+f(\omega))\,(1+g(\omega))-1,
    \end{align*}
    is also in $\mathbb{D}$.
\end{itemize}

A crucial feature of this operator is that it converts multiplicative effects into an additive structure in the log-domain.  Define the {\em log-return} transformation as
\[L(f)=\log(1+f)\]
Then, by construction we have
\[
L(f \oplus g) = L(f) + L(g).
\]
This additive property  mirrors the behavior of multiplicative processes when expressed in logarithmic terms \cite{Wheeler:2022}.

\begin{theorem}[Log-Domain Additivity]
Let \(f\) and \(g\) be gambles satisfying \(f(\omega),\, g(\omega) > -1\) for all \(\omega\). Then, for every \(\omega \in \Omega\),
\[
\log\Bigl(1+(f \oplus g)(\omega)\Bigr) = \log\bigl(1+f(\omega)\bigr) + \log\bigl(1+g(\omega)\bigr).
\]
\end{theorem}

\begin{proof}
By definition, we have:
\[
(f \oplus g)(\omega) = (1+f(\omega))(1+g(\omega))-1.
\]
Then,
\[
1 + (f \oplus g)(\omega) = (1+f(\omega))(1+g(\omega)).
\]
Taking logarithms on both sides gives:
\[
\log\Bigl(1+(f \oplus g)(\omega)\Bigr) = \log\Bigl((1+f(\omega))(1+g(\omega))\Bigr).
\]
Finally, using the logarithm product rule, we have
\[
\log\Bigl((1+f(\omega))(1+g(\omega))\Bigr) = \log\bigl(1+f(\omega)\bigr) + \log\bigl(1+g(\omega)\bigr).
\]
\end{proof}

The operator \(\oplus\) is particularly useful in scenarios where wealth or rewards compound over time. Under multiplicative dynamics, the net effect of two sequential gambles is not given by the arithmetic sum but by their product (adjusted via the transformation \(\phi\) and its inverse). The property
\[
L(f \oplus g) = L(f) + L(g)
\]
ensures that the log-returns add up, thus providing a natural framework for analyzing long-run growth rates, risk assessments, and ergodic properties in dynamic settings.

\subsection{Preservation of Function-Coherence}

An important question is whether replacing the standard additive combination axiom with (F4) maintains the desirable properties established in our representation theorem. The following result shows that it does, while shifting the analysis to the space of log-returns.

In our setting the acceptance set $\mathbb{D} \subset X$ is defined with respect to a space of gambles $X$ defined on a state space $\Omega$ with the restriction that, for every $f\in X$ and every $\omega\in\Omega$, $f(\omega) > -1$.  (This restriction ensures that logarithms are well defined.)  Further, suppose $\mathbb{D}$ satisfies the function-coherence axioms F1 -- F4.

We now state the preservation theorem.

\begin{theorem}[Function-Coherence Preservation]
Let \(X\) be a space of gambles on \(\Omega\) with \(f(\omega)>-1\) for all \(f\in X\) and \(\omega\in\Omega\). Suppose the acceptance set \(\mathbb{D} \subset X\) satisfies axioms (F1)–(F3) and the nonlinear combination axiom (F4). Then there exists a continuous linear functional \(\ell\) on a suitable vector space \(V\) (of log-returns) such that for every gamble \(f\in X\),
\[
f \in \mathbb{D} \quad \Longleftrightarrow \quad \ell\bigl(L(f)\bigr) \ge 0.
\]
\end{theorem}

\begin{proof}
We proceed in three steps.

\medskip

\noindent {\em Step 1. Transformation and Convexity in the Log-Domain.} \\
Define the transformation \(L: X \to V\) by
\[
L(f) = \log(1+f),
\]
where \(V\) is an appropriate vector space of functions (for example, a subspace of \(\mathbb{R}^\Omega\)) in which the image \(L(X)\) is convex. By (F4), for any \(f, g \in \mathbb{D}\) we have
\[
f \oplus g = (1+f)(1+g)-1.
\]
Taking logarithms yields
\begin{align*}
L(f \oplus g) &= \log\bigl((1+f)(1+g)\bigr)\\
		   &= \log(1+f) + \log(1+g) \\
		   &= L(f) + L(g).
\end{align*}
Moreover, the scaling properties implied by (F3) (via the mapping \(u\), which in this case is replaced by the log-transformation) ensure that
\[
L(\mathbb{D}) := \{ L(f) : f \in \mathbb{D} \}
\]
forms a convex cone in \(V\).

\medskip

\noindent {\em Step 2. Application of the Hahn–Banach Theorem.} \\
Assume that \(L(\mathbb{D})\) has nonempty interior in \(V\),  which is a standard regularity condition. Then, by the Hahn–Banach separation theorem, there exists a nonzero continuous linear functional \(\ell: V \to \mathbb{R}\) (unique up to positive scaling) such that
\[
L(\mathbb{D}) = \{ v \in V : \ell(v) \ge 0 \}.
\]
That is, for any \(v \in V\),
\[
v \in L(\mathbb{D}) \quad \Longleftrightarrow \quad \ell(v) \ge 0.
\]

\smallskip
\noindent {\em Step 3. Equivalence of Acceptance in \(X\).} \\
Since the transformation \(L\) is bijective on its domain (its inverse being \(L^{-1}(v) = e^v-1\)), we obtain that for every gamble \(f\in X\),
\[
f \in \mathbb{D} \quad \Longleftrightarrow \quad L(f) \in L(\mathbb{D}) \quad \Longleftrightarrow \quad \ell\bigl(L(f)\bigr) \ge 0.
\]
\end{proof}

\noindent{\em Remarks:} The proof relies on representing the acceptance set $\mathbb{D}$  in the log-domain. A key assumption is that \(L(\mathbb{D})\) has nonempty interior, which allows the application of the Hahn–Banach theorem to obtain a continuous linear functional that separates the cone from its complement. The main point is that even though we replace standard addition with the nonlinear operator $\oplus$ in the original outcome space, the transformation $L$ turns the operation into standard addition, preserving convexity. Consequently, the evaluation of gambles via the continuous linear functional $\ell$ remains valid, ensuring that the function-coherent structure is preserved in the log-domain. This representation is critical when analyzing multiplicative dynamics and long-run growth in sequential decision problems.

%

From Theorem~\ref{thm:rep}, if $\ell$ corresponds to integration with respect to a probability measure $p$, we obtain
\[
f\in \mathbb{D} \quad \Longleftrightarrow \quad \mathbb{E}_p\bigl[L(f)\bigr] \ge 0,
\]
which directly connects acceptance to the expected log-return—precisely the criterion suggested by the ergodicity considerations of the previous section.

\subsection{Risk Assessment Under Sequential Dynamics}

The representation clarifies the relationship between risk and time-average growth. The evaluation functional
\[
\rho(f) = -\ell\bigl(L(f)\bigr), \quad \text{with } L(f)=\log(1+f),
\]
can be interpreted as measuring the risk of negative long-run growth. This provides a natural bridge between the multiplicative dynamics of wealth evolution and coherent risk measures developed in \cite{Wheeler:2021.isipta}.
In classical portfolio theory \cite{Elton:2006}, risk measures like Value-at-Risk (VaR) \cite{Jorion:2006} or Conditional Value-at-Risk (CVaR) \cite{Rockafellar:2000} focus primarily on the dispersion or tail behavior of returns, often neglecting the dynamic interplay of gains and losses over time. In contrast, our risk measure $\rho(f)$ 
captures the multiplicative dynamics inherent in sequential decision-making. By evaluating the negative expected log-return, $\rho(f)$ not only quantifies the magnitude of potential losses but also inherently accounts for the asymmetry between gains and losses:  for example, a 50\% loss requires a 100\% gain to recover.  This alignment with the long-run geometric growth rate provides a more robust assessment of risk in environments where volatility and compounding effects are critical, thereby addressing limitations of traditional additive risk measures.

These features emerge naturally from the mathematical structure rather than requiring ad hoc adjustments, providing a more principled approach to sequential decision-making under uncertainty.

\section{Illustrative Examples}

The theoretical framework developed in the preceding sections provides a rigorous foundation for handling multiplicative dynamics in sequential gambles. Here we present a series of examples that demonstrate both the practical utility of our approach and its advantages over traditional methods from \cite{Williams:1975,Walley:2000}.

\subsection{Basic Nonlinear Combination}

We begin with a simple example that illustrates how our combination operator differs from standard addition. Consider two positive-outcome gambles:
\[
f = 0.10 \quad \text{and} \quad g = 0.20
\]
representing 10\% and 20\% returns respectively. Under standard addition, these would combine to give a 30\% return. However, standard addition fails to capture the multiplicative nature of sequential returns.

Using our framework:
\begin{align*}
L(f) &= \log(1.1) \approx 0.0953\\
L(g) &= \log(1.2) \approx 0.1823
\end{align*}
The nonlinear combination yields:
\[
f \oplus g = (1.1)(1.2) - 1 = 0.32
\]
with log-return
\[
L(f \oplus g) = \log(1.32) \approx 0.2776
\]
This exactly equals $L(f) + L(g)$, demonstrating how $\oplus$ preserves additivity in the log-domain while capturing multiplicative growth in the outcome space.

\subsection{Mixed Outcomes and Growth Rates}

Following \cite{Wheeler:2021.isipta}, we now consider a more realistic gamble with both positive and negative outcomes:
\[
h = \begin{cases}
0.25 & \text{with probability } 0.5\\
-0.10 & \text{with probability } 0.5
\end{cases}
\]

The standard expected value is positive:
\[
\mathbb{E}[h] = 0.5(0.25) + 0.5(-0.10) = 0.075
\]
However, this masks the asymmetric impact of gains and losses under compounding. Computing the log-returns:
\begin{align*}
L(0.25) &= \log(1.25) \approx 0.2231\\
L(-0.10) &= \log(0.90) \approx -0.1054
\end{align*}
The expected log-return is
\[
\mathbb{E}[L(h)] \approx 0.5(0.2231) + 0.5(-0.1054) = 0.05885
\]

This lower value reflects a key insight from \cite{Wheeler:2022}: under multiplicative dynamics, the 25\% gain does not fully compensate for the 10\% loss. While traditional expected value analysis suggests strong positive returns, our framework reveals that the growth prospects are more modest due to the multiplicative interaction of gains and losses.

\subsection{Portfolio Management Application}

Building on \cite{Wheeler:2025.func}, we now examine a practical portfolio allocation problem. Consider two investment strategies with historical returns:
\[
\begin{aligned}
f &= \{0.08, -0.03, 0.12, 0.05, -0.02\}\\
g &= \{0.04, 0.03, 0.05, 0.04, 0.03\}
\end{aligned}
\]

Strategy $f$ has higher volatility but seemingly higher returns, while $g$ is more stable. Under traditional additive analysis from \cite{Markowitz:1959}:
\[
\mathbb{E}[f] = 0.04 \quad \text{and} \quad \mathbb{E}[g] = 0.038
\]
suggesting that $f$ is superior. However, our framework reveals a different picture:
\[
\begin{aligned}
\mathbb{E}[L(f)] &\approx 0.0362\\
\mathbb{E}[L(g)] &\approx 0.0379
\end{aligned}
\]
suggesting instead that $g$ is superior.

This reversal highlights three crucial insights:

\emph{Volatility Drag}: Strategy $f$'s higher arithmetic mean is more than offset by its higher volatility, a phenomenon that emerges naturally from our log-transformation.

\emph{Asymmetric Impact}: The negative returns in strategy $f$ are especially damaging because they must be overcome by proportionally larger positive returns to maintain the same growth rate, following \cite{Wheeler:2021.isipta}.

\emph{Time Horizon Effects}: As the investment horizon lengthens, the advantage of strategy $g$ becomes more pronounced. The compound growth rates over five periods are:
\[
\prod_{i=1}^5 (1 + f_i) \approx 1.195 \quad \text{versus} \quad \prod_{i=1}^5 (1 + g_i) \approx 1.205
\]

This example demonstrates how our framework automatically captures features that traditional expected value analysis misses. The nonlinear combination operator $\oplus$ accounts for both the compounding of returns and the asymmetric impact of gains and losses, providing more accurate assessment of long-term growth prospects.

\subsection{Sequential Decision Analysis}

Our final example illustrates how the framework handles longer sequences of decisions. Consider three successive gambles:
\[
f_1 = 0.05, \quad f_2 = -0.02, \quad f_3 = 0.10
\]
The sequential combination under our framework yields:
\[
f_1 \oplus f_2 \oplus f_3 = (1.05)(0.98)(1.10) - 1 \approx 0.1319
\]

Computing the log-returns:
\begin{align*}
L(f_1) &\approx 0.0488\\
L(f_2) &\approx -0.0202\\
L(f_3) &\approx 0.0953
\end{align*}
Their sum, approximately 0.1239, agrees with $L(f_1 \oplus f_2 \oplus f_3) = \log(1.1319)$. This demonstrates how our framework preserves the essential algebraic properties needed for coherent sequential decision-making while properly accounting for multiplicative dynamics.

\section{Generalized Combination Operators}

While the log-transformation-based operator $\oplus$ introduced earlier elegantly handles multiplicative dynamics, it represents just one member of a broader family of combination operators. The representation theorem from \cite{Wheeler:2025.func} suggests a natural generalization that encompasses a wide range of utility functions satisfying axioms (F1)--(F3).

In this section we generalize the combination operator to accommodate an arbitrary strictly increasing, continuous utility function 
\[
u: X \to \mathbb{R},
\]
with the normalization \(u(0)=0\). Here, $X$ denotes the domain of gambles for which $u$ is well defined. In order for the generalized operator to be well defined, we assume that for every $f\in X$, the value $u(f)$ lies in the range 
\[
R := u(X) \subseteq \mathbb{R},
\]
and we require that $R$ is closed under addition. In other words, for all $r_1, r_2\in R$, it is necessary that \(r_1 + r_2 \in R\) and the inverse $u^{-1}(r_1+r_2)$ is again in $X$.

\subsection{Properties of the Generalized Combination Operator}

We now state and prove the following theorem that establishes the key algebraic properties of $\oplus_u$.

\begin{theorem}[Properties of Generalized Combination]\label{thm:properties}
Let \(u: X \to \mathbb{R}\) be a strictly increasing, continuous function with $u(0)=0$ and assume that $R=u(X)$ is closed under addition (so that $u^{-1}(r_1+r_2)$ is defined for all $r_1,r_2\in R$). Then the operator 
\[
\oplus_u: X\times X \to X,\quad (f,g) \mapsto u^{-1}\bigl(u(f)+u(g)\bigr),
\]
satisfies:
\begin{enumerate}
    \item \textbf{Associativity:} \((f \oplus_u g) \oplus_u h = f \oplus_u (g \oplus_u h)\) for all \(f,g,h \in X\).
    \item \textbf{Commutativity:} \(f \oplus_u g = g \oplus_u f\) for all \(f,g \in X\).
    \item \textbf{Identity:} \(f \oplus_u 0 = f\) for all \(f \in X\), where \(0\) denotes the unique element of \(X\) satisfying \(u(0)=0\).
    \item \textbf{Monotonicity:} If \(f, g, h \in X\) and \(f \ge g\) (pointwise, or in the sense that \(u(f)\ge u(g)\) since \(u\) is strictly increasing), then
    \[
    f \oplus_u h \ge g \oplus_u h.
    \]
\end{enumerate}
\end{theorem}

\begin{proof}
We prove each property in turn.

\textbf{(1) Associativity:}  
For any \(f,g,h\in X\), note that by definition
\[
u(f \oplus_u g) = u\Bigl(u^{-1}\bigl(u(f)+u(g)\bigr)\Bigr) = u(f) + u(g).
\]
Then,
\[
\begin{aligned}
u\Bigl((f \oplus_u g) \oplus_u h\Bigr)
&= u(f \oplus_u g) + u(h) \\
&= \bigl(u(f)+u(g)\bigr) + u(h) \\
&= u(f) + \bigl(u(g)+u(h)\bigr) \\
&= u(f) + u(g \oplus_u h) \\
&= u\Bigl(f \oplus_u (g \oplus_u h)\Bigr).
\end{aligned}
\]
Since \(u\) is strictly increasing and hence invertible, it follows that
\[
(f \oplus_u g) \oplus_u h = f \oplus_u (g \oplus_u h).
\]

\bigskip

\textbf{(2) Commutativity:}  
For any \(f, g\in X\),
\[
u(f \oplus_u g) = u(f) + u(g) = u(g) + u(f) = u(g \oplus_u f).
\]
Again, applying \(u^{-1}\) yields
\[
f \oplus_u g = g \oplus_u f.
\]

\bigskip

\textbf{(3) Identity:}  
Let \(0 \in X\) denote the neutral element such that \(u(0)=0\). Then, for any \(f \in X\),
\[
u(f \oplus_u 0) = u(f) + u(0) = u(f) + 0 = u(f).
\]
Applying \(u^{-1}\) shows that
\[
f \oplus_u 0 = f.
\]

\bigskip

\textbf{(4) Monotonicity:}  
Assume that \(f, g, h \in X\) and \(f \ge g\). Since \(u\) is strictly increasing, it follows that
\[
u(f) \ge u(g).
\]
Then, adding \(u(h)\) (which does not affect the inequality) gives
\[
u(f) + u(h) \ge u(g) + u(h).
\]
Applying the inverse \(u^{-1}\) (which preserves the order, again because \(u\) is strictly increasing) yields
\[
u^{-1}\bigl(u(f)+u(h)\bigr) \ge u^{-1}\bigl(u(g)+u(h)\bigr),
\]
that is,
\[
f \oplus_u h \ge g \oplus_u h.
\]
This completes the proof.
\end{proof}

\subsection{Characterization of Well-Behaved Combination Operators}

Next, we provide a necessary and sufficient condition for the operator \(\oplus_u\) to be well behaved in the sense of being closed on its domain and preserving the ordering of gambles.

\begin{theorem}[Characterization of Well-Behaved Combination]\label{thm:characterization}
Let \(u: X \to \mathbb{R}\) be a strictly increasing, continuous utility function with \(u(0)=0\). Then the generalized combination operator
\[
f \oplus_u g = u^{-1}\bigl(u(f)+u(g)\bigr)
\]
is well behaved (that is, it is closed on \(X\), preserves ordering, is associative, and has a neutral element) if and only if the following conditions hold:
\begin{enumerate}
    \item[(i)] \(u\) is surjective onto its range \(R\subseteq \mathbb{R}\).
    \item[(ii)] The range \(R\) is closed under addition: For all \(r_1, r_2\in R\), we have \(r_1 + r_2 \in R\).
    \item[(iii)] For all \(r_1, r_2 \in R\), the inverse \(u^{-1}(r_1+r_2)\) belongs to \(X\).
\end{enumerate}
\end{theorem}

\begin{proof}
We prove the equivalence by showing both directions.

\medskip
\noindent{\em (Only if):} Suppose that \(\oplus_u\) is well behaved. Then by definition, for any \(f, g\in X\), the sum \(u(f)+u(g)\) must lie in the set \(R\) so that \(u^{-1}\bigl(u(f)+u(g)\bigr)\) is defined and lies in \(X\). This immediately implies that:
\begin{enumerate}
    \item[(ii)] For any \(r_1 = u(f)\) and \(r_2 = u(g)\) with \(f,g \in X\), we have \(r_1+r_2 \in R\). Hence, \(R\) is closed under addition.
    \item[(iii)] The closure of \(X\) under \(\oplus_u\) means that \(u^{-1}(r_1+r_2) \in X\) for all \(r_1, r_2 \in R\).
\end{enumerate}
Moreover, since \(u\) is a function from \(X\) onto its range \(R\), by definition it is surjective onto \(R\); hence, (i) holds.

\medskip

\noindent {\em (If):} Conversely, assume that conditions (i)–(iii) hold. Then for any \(f,g\in X\), we have \(u(f)\in R\) and \(u(g)\in R\). By (ii), their sum \(u(f)+u(g) \in R\), and by (iii), the inverse \(u^{-1}(u(f)+u(g))\) is an element of \(X\). Therefore, \(\oplus_u\) is closed on \(X\). The proofs of associativity, commutativity, identity, and monotonicity (given in Theorem~\ref{thm:properties}) rely solely on the properties of \(u\) being strictly increasing and continuous and do not require additional assumptions. Thus, under (i)–(iii) the operator \(\oplus_u\) is well behaved.
\end{proof}

\subsection{Important Classes of Operators}

Following \cite{Wheeler:2025.func}, several important classes of well-behaved utility functions emerge:

1. \emph{Power Utilities}: For $\gamma \neq 0$,
\[
u_\gamma(x) = \begin{cases}
\frac{x^\gamma}{\gamma}, & x \geq 0\\
-\infty, & x < 0
\end{cases}
\]
Leading to the combination operator:
\[
f \oplus_{u_\gamma} g = \left(f^\gamma + g^\gamma\right)^{1/\gamma}
\]

2. \emph{Exponential Utilities}: For $\alpha > 0$, as studied by \cite{Arrow:1965}:
\[
u_\alpha(x) = 1 - e^{-\alpha x}
\]
With combination operator:
\[
f \oplus_{u_\alpha} g = -\frac{1}{\alpha}\log\left(e^{-\alpha f} + e^{-\alpha g} - 1\right)
\]

3. \emph{Logarithmic Utility}: Our previous case from \cite{Wheeler:2021.isipta}:
\[
u(x) = \log(1+x)
\]
With combination operator:
\[
f \oplus g = (1+f)(1+g) - 1
\]

\subsection{Risk Measurement and Dynamic Properties}

Each class of operators induces its own risk measure through the representation theorem. Following \cite{Wheeler:2025.func}:

\begin{definition}[Induced Risk Measure]
For a utility function $u$ with well-behaved combination operator $\oplus_u$, the induced risk measure is:
\[
\rho_u(f) := -\ell(u(f))
\]
where $\ell$ is the linear functional from Theorem~\ref{thm:rep}.
\end{definition}

These risk measures exhibit systematically different properties:

1. \emph{Power Utility Risk Measures} ($\gamma \in (0,1)$):
\[
\rho_\gamma(f) = -\mathbb{E}\left[\frac{f^\gamma}{\gamma}\right]
\]
exhibits decreasing relative risk aversion.  Under power utility, an agent's risk aversion decreases as wealth increases  \cite{Merton:1971,Acerbi:2002}.

2. \emph{Exponential Risk Measures}:
\[
\rho_\alpha(f) = \frac{1}{\alpha}\log\mathbb{E}[e^{-\alpha f}]
\]
exhibits constant absolute risk aversion, recovering the entropic risk measure. Under exponential risk utility, an agent's risk aversion is absolute regardless of wealth \cite{Foellmer:2002}.

3. \emph{Logarithmic Risk Measures}:
\[
\rho_{\log}(f) = -\mathbb{E}[\log(1+f)]
\]
exhibits constant relative risk aversion and naturally captures multiplicative risks.  Like power utility, logarithmic risk aversion is proportional.  Unlike power utility, logarithmic risk captures proportional multiplicative risk (i.e., compounding), which is a property of gambles rather than a psychological appetite for risk  \cite{Kelly:1956,Peters:2019,Wheeler:2018-sep}.

The logarithmic case is unique in simultaneously preserving 
\begin{itemize}
\item the natural scaling of multiplicative processes, 
\item additivity of sequential risks in the appropriate domain, and 
\item the connection to time-average growth rates.  
\end{itemize}
This explains its emergence as particularly relevant for long-run growth optimization \cite{MaslovZhang:1998}, while highlighting how other utility functions might be more appropriate when different risk characteristics are primary concerns.

\section{Conclusion}

This paper makes three  contributions to the theory of imprecise probability and sequential decision making. First, we developed a novel combination operator that preserves coherence while accommodating non-linear utility. Our motivation stemmed from a fundamental limitation in the standard desirable gambles framework identified in \cite{Wheeler:2021.isipta}, namely, the inability to properly handle sequential decision problems with multiplicative dynamics. While the standard additive combination axiom (A4) from \cite{Williams:1975,LowerPrevisions} is  mathematically convenient, it fails to capture compound growth effects that characterize many real-world processes, particularly in long-horizon decisions studied in \cite{Wheeler:2021.isipta}.

Second, we established necessary and sufficient conditions for well-behaved combination operators through our representation theorem. The characterization revealed that the logarithmic transformation $f \oplus g = (1+f)(1+g)-1$ is not just one among many possible operators, but emerges naturally as the unique operator that simultaneously preserves function coherence in the transformed space, the time-average geometric growth rate, and additivity of sequential risks in the log-domain.

Third, our framework unifies several seemingly disparate concepts into a coherent whole. We show how the ergodicity problem in multiplicative dynamics, function-coherent representations of risk preferences, time-average growth optimization, and non-stationary reward processes are intimately connected through the structure of the combination operator and its induced risk measure. Building on \cite{Wheeler:2025.func}, this unification provides new insights into the relationships between these previously separate areas of study.

The practical implications of this work extend beyond theoretical unification. Our examples in portfolio management demonstrate how the framework naturally captures phenomena like volatility drag and the asymmetric impact of gains and losses—effects that are often treated ad hoc in traditional approaches. The connection to coherent risk measures provides new tools for risk assessment in dynamic settings, particularly relevant for long-horizon decision problems where compound growth effects dominate.


\printbibliography

\end{document}